\documentclass[%
    11pt,
    a4paper,
    DIV=12,
    BCOR=0mm,
    numbers=endperiod,
    twoside=semi,
    headings=small
]{scrartcl}

\usepackage[automark,headsepline]{scrpage2} 
\newcommand{\quoteparagraph}[1]{\noindent{\normalsize \bfseries \sffamily #1}\hspace{1em}}

\usepackage{etex}
\usepackage[utf8]{inputenc}
\usepackage[english]{babel}
\usepackage[T1]{fontenc}
\usepackage{graphicx}
\usepackage{subfig}
\usepackage{amsmath}
\usepackage{amssymb}
\usepackage{amsthm}
\usepackage{thmtools}
\usepackage{mathtools}
\usepackage{enumerate}
\usepackage{listings}
\usepackage{multicol}
\usepackage{stmaryrd}
\usepackage[arrow, matrix, curve]{xy}
\usepackage{framed}
\usepackage{leftidx}
\usepackage{hyperref}
\usepackage{float}

\usepackage{tikz}
\usepackage{circuitikz}
\usepackage{pgfplots}
\usepackage{ifthen}
\usetikzlibrary{arrows,automata,shapes,calc}
\tikzstyle{e}=[
        ->,
        semithick,
        node distance=2.5cm,
        auto,
        >=stealth',
        initial text={}
]

\newcommand{\candle}[5]{
  \def\candlewidth{0.3};
  \draw[fill=white,thick] (#1,#4) -- (#1,#5);
  \ifthenelse { \lengthtest{#2 pt < #3 pt} }
    {\draw[thick,fill=white!100] (#1-\candlewidth/2,#2) rectangle (#1+\candlewidth/2,#3);}
    {\draw[thick,fill=black!50 ] (#1-\candlewidth/2,#2) rectangle (#1+\candlewidth/2,#3);}
}

\newcommand{\candlepath}[1]{
  \def\width{1};

  \def\N{0};
  \foreach \x[count=\xi from 0,remember=\xi as \N] in {#1} {}

  \pgfmathsetmacro\h{\width/\N};

  \draw (0,0) \foreach \x[count=\xi from 0] in {#1} { \ifnum \xi>0 -- \fi (\h*\xi,\x) };

  \candlefrompath{#1}
}

\newcommand{\candlefrompath}[1]{
  \def\N{0};
  \def\O{0};
  \def\C{0};
  \pgfmathsetmacro{\H}{max(#1)};
  \pgfmathsetmacro{\L}{min(#1)};
  \foreach \x[remember=\x as \C] in {#1} {}
  \foreach \x[count=\xi from 0,remember=\x as \O] in {#1} {
    \ifnum \xi<1
      \breakforeach
    \fi
  }

  \candle{-0.3}{\O}{\C}{\H}{\L}
}

\DeclareSymbolFont{bbold}{U}{bbold}{m}{n}
\DeclareSymbolFontAlphabet{\mathbbold}{bbold}

\newtheorem{thm}{Theorem}[section]
\newtheorem{definition}[thm]{Definition}
\newtheorem{cor}[thm]{Corollary}
\newtheorem{sample}[thm]{Example}
\newtheorem{observation}[thm]{Observation}
\newtheorem{remark}[thm]{Remark}
\newtheorem{remarks}[thm]{Remarks}
\newtheorem{assumptions}[thm]{Assumptions}

\makeatletter
\g@addto@macro\bfseries{\boldmath}
\makeatother

\newcommand{\IN}{\mathbb{N}}
\newcommand{\IR}{\mathbb{R}}

\newcommand{\im}{\operatorname{im}}
\newcommand{\inv}{^{-1}}
\newcommand{\mf}[1]{\mathfrak{#1}}
\newcommand{\mc}[1]{\mathcal{#1}}
\DeclarePairedDelimiter{\abs}{\lvert}{\rvert}
\newcommand{\set}[1]{\left\{#1\right\}}

\DeclareTextFontCommand{\emph}{\normalfont\bfseries}

\hypersetup{
  pdftitle    = {Correctness of Backtest Engines},
  pdfauthor   = {Robert Löw, Stanislaus Maier-Paape and Andreas Platen},
  pdfkeywords = {backtest evaluation, historical simulation, trading system, candle chart, imperfect data, price model, correctness test, backtest correctness},
  pdfborder={0 0 0.5}
}

\author{
  \normalsize \textsc{Robert Löw}\\[-0.3em]
    \small \textit{Institut f\"ur Mathematik, RWTH Aachen,}\\[-0.6em]
    \small \textit{Templergraben 55, D-52052 Aachen, Germany}\\[-0.6em]
    \small \href{mailto:robert.loew@rwth-aachen.de}{robert.loew@rwth-aachen.de}\\
  \\[-0.75em]
  \normalsize \textsc{Stanislaus Maier-Paape}\\[-0.3em]
    \small \textit{Institut f\"ur Mathematik, RWTH Aachen,}\\[-0.6em]
    \small \textit{Templergraben 55, D-52052 Aachen, Germany}\\[-0.6em]
    \small \href{mailto:maier@instmath.rwth-aachen.de}{maier@instmath.rwth-aachen.de}\\
  \\[-0.75em]
  \normalsize \textsc{Andreas Platen}\\[-0.3em]
    \small \textit{Institut f\"ur Mathematik, RWTH Aachen,}\\[-0.6em]
    \small \textit{Templergraben 55, D-52052 Aachen, Germany}\\[-0.6em]
    \small \href{mailto:platen@instmath.rwth-aachen.de}{platen@instmath.rwth-aachen.de}
}
\date{
  \vspace{0.25em}
  \normalsize\today
  \vspace{-7mm}
}
\title{
  \vspace{-2cm}
  \Large Correctness of Backtest Engines
  \vspace{-5mm}
}

\clearscrheadfoot
\cehead{\normalfont\small R. Löw, S. Maier-Paape, and A. Platen}  
\cohead{\normalfont\small Correctness of Backtest Engines}        
\begin{document}

	\maketitle
	\hrule

	\begin{quote}
	\small
	\quoteparagraph{Abstract}
	In recent years several trading platforms appeared which provide a backtest engine to calculate historic performance of self designed trading strategies on underlying candle data. The construction of a correct working backtest engine is, however, a subtle task as shown by Maier-Paape and Platen (cf. [\href{http://arxiv.org/abs/1412.5558}{arXiv:1412.5558}, q-fin.TR]). Several platforms are struggling on the correctness.

	In this work, we discuss the problem how the correctness of backtest engines can be verified.
	We provide models for candles and for intra-period prices which will be applied to conduct a proof of correctness for a given backtest engine if the here provided tests on specific model candles are successful.
	Furthermore, we hint to algorithmic considerations in order to allow for a fast implementation of these tests necessary for the proof of correctness.

	\quoteparagraph{Keywords} backtest evaluation, historical simulation, trading system, candle chart, imperfect data, price model, correctness test, backtest correctness

	\quoteparagraph{JEL classification}
	C15, 
	C63, 
	C88, 
	C99  

	\end{quote}


	\section{Introduction}

	The use of trading software allows for programming trading strategies which are then automatically executed.
	Typically, software solutions for this task also provide a backtest engine which evaluates the user-written scripts on historical data and provides the user with performance values, entries and exits of positions and further data which would have occurred if the trading script was used in the past.
	In general, the historical data is not exact in the sense that the price at a specific time is not known. Instead, prices are summed up for time periods in so called candle data or bar data which consists of the maximum price ($high$), the minimum price ($low$), the price at the beginning of the period ($open$) and the price at the end of the period ($close$) for a given period length (e.g. one candle per day).
	Due to this lack of information the unique decidability of several combinations of orders is lost when only candle data are available, see e.g. \cite{MP2014} and also the book of Pardo~\cite[Chapter 6, Section ``Software Limitations'']{Pardo2008} or Harris~\cite[Chapter 6]{Harris2008}.

	Before going into the details of our backtest correctness testing model, it is worth noting that, besides the above mentioned problems, there are also other limitations of backtesting which cannot be neglected, see e.g. the books of Chan~\cite[Chapter 3]{Chan2009}, Pardo~\cite[Chapter 6]{Pardo2008} and Harris~\cite[Chapter 6]{Harris2008} and for trading options the book of Izraylevich and Tsudikman~\cite[Chapter 5]{IT2012}. Since a backtest just simulates the behavior of a trading system over the past it is strongly limited in predicting the future. However, it is common to optimize a parameter dependent strategy over the historical data to maximize some objective function. Such an optimization process can rapidly get very time consuming. Therefore to efficiently backtest a trading system for different parameter choices Ni and Zhang~\cite{NZ2005} presented a method, but they do not explain the backtest evaluation itself. Computing such an ``optimal'' parameter setting does not ensure optimal parameters for the future but can even lead to tremendous losses. Here one often speaks of backtest overfitting, see \cite{BBP+2014,BBP+2014a,CP2014} and also \cite[Chapter 6]{Pardo1992} for a detailed discussion. Therefore even a correct backtest engine needs to be applied carefully. Nevertheless it gives important information about a trading strategy.

	The aim of this paper, however, is not about a ``right'' application of backtest results, but on how to verify whether or not the backtests are performed correctly in the first place.

	One instance where common software solutions often struggle is the non-uniqueness of correct results of backtest engines.
	Almost all of them do not consider this issue.
	Furthermore, even if situations are obviously uniquely decidable, severals platforms sometimes have incorrect evaluations of their backtest engines.
	It is, however, of great importance for the users to have reliable correct backtest evaluations.
	Therefore Maier-Paape and Platen~\cite{MP2014} asked how backtest engines should decide for certain given standard order setups like limit and stop entry orders combined with typical intra-period stop or target exit orders, when only candle data are available.
	They could limit these examinations to single candles, arguing that any trade can be split up into several candles with different active orders.
	For the order combinations discussed, they provided decision trees with which for a given candle, depending on open, close, high and low value, the correct behavior of a backtest engine can be determined.
	They also introduced decision modes which provide deterministic rules for situations which are not uniquely decidable due to the lack of information in candle data.
	However, due to the numerous possibilities which can occur in different situations, it is hard to verify whether a given backtest engine of some trading platform is calculating its backtests correctly or not.
	In this paper, we therefore provide means for an algorithmic approach to this problem.

	In order to check whether a given backtest engine behaves according to the requirements presented in \cite{MP2014}, it is necessary to design test data.
	That means that we are looking for a finite number of so called ``test candles''.
	After checking the backtest evaluation on these finitely many candles, we want to be able to conclude general correctness of the backtest engine under reasonable assumptions.

	This work is divided into five sections.
	In Section~\ref{sec:setup} we discuss the preliminaries under which our examinations are conducted and formalize the problem discussed.
	The main result of Section~\ref{sec:correctness} is the development of a concept of ``model candles'', which allows for a proof of correctness of backtest engines in specific situations.
	In Section~\ref{sec:IPFs}, a model for intra-period prices will be introduced which allows to obtain the desired condition that guarantees the completeness of all model candles and all corresponding possible outcomes of a backtest.
	We end this paper in Section~\ref{sec:conclusion} with our conclusions.


	\section{Problem statement and preliminaries}
	\label{sec:setup}

	First, we want to specify the situation which we examine.
	In Section~\ref{subsec:assumptions} we give some assumptions needed to make the decisions for a backtest unique. Since we can reduce the problem of testing for correctness to one single candle for each situation we define the setup for such situations in Section~\ref{subsec:setup_ipf}. The concept of a backtest engine is explained in Section~\ref{subsec:results}.

	\subsection{Assumptions}\label{subsec:assumptions}

	At this point, we will not make specific assumptions about the candles which we examine.
	Clearly, every candle satisfies
	\[
		low\leq\min\set{open,close}\leq\max\set{open,close}\leq high.
	\]
	Due to lack of information caused by the candle data, some general assumptions have to be made.

	\begin{assumptions}[see {\cite[Section 2.1]{MP2014}}]~\label{assum}
	\begin{itemize}
		\item No intra-period gaps: We assume a continuous intra-period price development during a candle.
		\item Market liquidity: All orders are filled at the requested price and thus without slippage.
		\item Local worst case/best case: Worst case and best case decisions are determined locally, i.e. the profit of a trade is evaluated as if it were closed on the close of the candle.
 		\item Local decision: The backtest engine decides only by considering the candle in question.
			Previous and future candles do not influence the decision made.
	\end{itemize}
	\end{assumptions}

	\begin{remark}
		Some of these assumptions are quite unreasonable for realistic price charts and markets.
		However, since the backtest engine has only candle data at hand, these assumptions are necessary and reasonable to decide a possible outcome of these market situations.
	\end{remark}

	To ease the understanding of the concepts presented, we will use an example throughout this work.

	\begin{sample}\label{bsp0}
		At the beginning of a candle, we might have no open position but two active orders, e.g. a stop buy entry order at level $53$, which enters a long position if the price reaches $53$, and a protective stop loss order at level $51$, which exits the long position (but only if it is opened beforehand) if the price falls below $51$.
		Given an arbitrary candle, the backtest engine should now decide which of these orders is executed and whether or not there are open positions at the end of the candle.
	\end{sample}

	\subsection{The setup and intra period prices}\label{subsec:setup_ipf}

	We can formalize which orders are active and what the position is prior to the candle.
	\begin{definition}
		Suppose we have $m\in\IN$ orders at levels $L_1<\ldots<L_m$, respectively, with $L_i\in\IR$ rounded to tick size of the asset for all $1\leq i\leq m$. Let $p\in\set{-1,0,1}$ denote the position status which was present prior to the candle in question, where $0$ stands for a flat position, $-1$ for a short position and $1$ for a long position.
		A combination of such $m$ orders and the position type $p$ is called a \emph{setup}.
	\end{definition}

	In the following we only consider setups which allow for at most one entry execution and one exit execution.
	The available entry orders are limit buy/sell and stop buy/sell, i.e. ``EnterLongLimit'', ``EnterShortLimit'', ``EnterLongStop'' and ``EnterShortStop''.
	As exit orders we use stop loss and profit target (cf. \cite[Section 2.2, 2.3 and 2.5]{MP2014} for the relevant decision trees).
	Note that we neglect market orders for entry and exit since they are trivially decidable.
	We do not investigate whether the results presented can be generalized to setups which allow for at most $a$ entry and $a$ exit executions, where $1<a\in\IN$, but analogous concepts might be applicable.
	Furthermore, we exclude the trivial case $m=0$ in which no order is active.

	Now, we can formalize Example~\ref{bsp0}.
	\begin{sample}\label{bsp1}
		We use the combination of a stop buy entry (EnterLongStop) order at level $L_2=53$ and a protective stop loss exit order at level $L_1=51$ attached to that entry order with $p=0$ as an example setup.
	\end{sample}

	For technical reasons only, the intra-period prices of a candle will be assumed to be continuous as follows, where we denote the set of all continuous functions $g:A\rightarrow B$ by $\mc{C}(A,B)$ and define $\IR^+:=\{x\in\IR \mid x\geq 0\}$.
	\begin{definition}
		We call $f\in \mc{C}([a,b],\IR^+)$ with $a,b\in\IR$ and $a<b$ an \emph{intra-period price function (IPF)}.
		By
		\[
			C(f):=\left(f(a),f(b),\max_{t\in[a,b]}f(t),\min_{t\in[a,b]}f(t)\right)\in(\IR^+)^4
		\]
		we denote the corresponding candle of the IPF $f$.
	\end{definition}
	Here, $a$ and $b$ are interpreted as the starting time and ending time of a period which is summed up to a candle and $f(t)$ as the price at time $t$ for $a\leq t\leq b$. Obviously, for each candle $c\in(\IR^+)^4$ one can construct an IFP $f$ such that $C(f)=c$.

	\subsection{Backtest engines and results}\label{subsec:results}

	Given an intra-period price function, we can define its result.
	\begin{definition}\label{def:CR}
		We define the \emph{result} of an IPF $f$ with respect to a given setup as the combination of~$(entry_f, exit_f)$, where $entry$ and $exit$ result from the correct application of the orders of the given setup.
		If $entry$ or $exit$ do not occur, they are denoted by value $-1$, respectively.
		We denote the result of the IPF $f$ by $R(f)$.
		Formally, for a given setup, this can be regarded as function
		\begin{equation}\label{eq:result_function}
			R:\bigcup_{a,b\in\IR,a<b}\mc{C}([a,b],\IR^+)\to(\IR^+\cup\set{-1})^2.
		\end{equation}

		With this, the \emph{combined candle data and result (CR)} of $f$ with respect to a given setup can be defined as the combination $CR(f):=(C(f),R(f))$.
	\end{definition}
	\begin{remarks}\label{rem:multiple_results}
		Note that results and CRs are defined depending on IPFs, which means that all CRs correspond to intra-period prices which fulfill the first two assumptions made in Assumptions~\ref{assum}.

		For given setup and candle data $(open,close,high,low)$, finding all possible results $(entry,\linebreak[0] exit)$ can be reformulated as finding $\im(CR)\cap\left(\set{(open,close,high,low)}\times\IR^2\right)$.
		Note that this set can have more than just one element, which means that the result for this candle is not unique.
	\end{remarks}

	For a given fixed setup, we can formalize the concept of a backtest engine used in this work.
	As we want to examine the behavior of a backtest engine on candle data, we can take such a candle as input data.
	The backtest engine should return a possible entry price at which a position may be opened and a possible exit price at which the position may be closed.
	Of course, entry or exit might not occur or might not have a price which can be uniquely determined (cf. Remark~\ref{rem:multiple_results}).
	Thus, the backtest engine also needs a ``backtest decision mode''~(short: backtest mode) which makes the entry and exit price unique.
	\begin{definition}\label{engine}
		For $BM:=\set{best~case,worst~case,ignore}$ we interpret a value in $(\IR^+)^4\times BM$ as the combination of the candle data $(open,close,high,low)\in(\IR^+)^4$ and the backtest mode according to Assumptions~\ref{assum}.

		For a given setup, we call a mapping
		\[
			E:(\IR^+)^4\times BM\to(\IR^+\cup\set{-1})^2
		\]
		a \emph{backtest engine}.

		The result $E(c_0,\mf{M})\in(\IR^+\cup\set{-1})^2$ for a candle $c_0\in(\IR^+)^4$ and a backtest mode $\mf{M}\in BM$ is interpreted as the combination of entry price and exit price, respectively, where again $-1$ denotes no entry/no exit.
	\end{definition}

	\begin{remark}
		On a particular chart, here a candle chart, the backtest engine of a platform decides the outcomes of a given strategy.
		It has to simulate scripts which are typically executed at the end of each candle and place orders.
		Furthermore, it has to keep track of placed orders and open positions.
		When all orders placed before a candle, which form the setup, are known, the backtest engine has to decide which of these orders are executed.
		Here, we consider only this last part of a backtest engine.
		So for us, a backtest engine $E$ has only the candle data (and the setup) as input, whereas the result function $R$ from \eqref{eq:result_function} requires an IPF, i.e. intra-period data.
		Thus, $E$ has to find results which correspond to the given candle, i.e. for a given candle $c_0$ the backtest engine has to find a result $r_0\in(\IR^+\cup\set{-1})^2$ s.t. $r_0=R(f)$ and $c_0=C(f)$ for an appropriate IPF $f$.
		The backtest mode specifies which IPF has to be chosen if there are several results possible corresponding to the given candle.
	\end{remark}

	Continuing Example~\ref{bsp1}, we show these concepts in application.
	\begin{sample}\label{bsp2}
		Using the setup from Example~\ref{bsp1} (stop buy entry order at $L_2=53$, stop loss exit order at $L_1=51$ with a flat position before the candle, i.e. $p=0$) and the IPF
		\[
			f:\left[0,\frac{5}{2}\pi\right]\to\IR^+, t\mapsto\sin(t)+52,
		\]
		(see Figure~\ref{figSinA}) we obtain the CR
		\[
			(C(f),R(f))=(open=52, close=53, high=53, low=51, entry_f=53, exit_f=51).
		\]

		In this example, we can see the non-uniqueness of results, if only a candle is given, by considering the IPF
		\[
			g:\left[0,\frac{3}{2}\pi\right]\to\IR^+, t\mapsto-\sin(t)+52,
		\]
		(see Figure~\ref{figSinB}) with the same resulting candle as for $f$, i.e. $C(f)=C(g)$, but the result
		\[
			R(g)=(entry_g=53, exit_g=-1)\neq R(f).
		\]
		Here, $g$ corresponds to the best case for this candle and $f$ to the worst case, as with a fictional $close$ of the candle we have for the cash value of $f$ and $g$ that
		\[
			close-entry_g=close-53=53-53=0>-2=51-53=exit_f-entry_f.
		\]
		At this point, we actually do not know whether there could be other cases, but for this setup there are at most $2$ possible results for a given candle. Choosing a backtest mode $\mf{M}\in BM$ then makes the decision for the backtest engine unique and allows to either choose the result from $f$, or from $g$ or also to ignore this trade.
	\end{sample}
	\begin{figure}

	  \subfloat[IPF $f$.]{
            \begin{tikzpicture}
                    \draw[color=black!0] (-3pt,-1) -- (-3pt,-1) node[anchor=east,color=black] {$51$}; 
                    \draw[color=black!0] (-3pt,1) -- (-3pt,1) node[anchor=east,color=black] {$53$}; 

                    \draw[color=black] (0,-1) -- (0,-1.3) node[anchor=north,color=black] {$0$}; 
                    \draw[color=black] (7.85,-1) -- (7.85,-1.3) node[anchor=north,color=black] {$\frac{5}{2}\pi$}; 

                    \draw[->,color=black!100] (0,-1.3) -- coordinate (y axis mid) (0,1.3); 
                    \draw[-,color=black!40] (0,1) -- coordinate (x axis mid) (8.3,1); 
                    \draw[-,color=black!40] (0,-1) -- coordinate (x axis mid) (8.3,-1); 

                    \draw[color=black!0] (6.5,0) -- (6.5,0) node[anchor=west,color=black] {$f$}; 

                    \draw[color=black,domain=0:7.85,samples=100]   plot (\x,{sin(\x r)})   node[right]{};
                    \candle{8.1}{0}{1}{1}{-1};
            \end{tikzpicture}
            \label{figSinA}
          }
          \hfill
          \subfloat[IPF $g$.]{
            \begin{tikzpicture}
                    \draw[color=black!0] (-3pt,-1) -- (-3pt,-1) node[anchor=east,color=black] {$51$}; 
                    \draw[color=black!0] (-3pt,1) -- (-3pt,1) node[anchor=east,color=black] {$53$}; 

                    \draw[color=black] (0,-1) -- (0,-1.3) node[anchor=north,color=black] {$0$}; 
                    \draw[color=black] (4.71,-1) -- (4.71,-1.3) node[anchor=north,color=black] {$\frac{3}{2}\pi$}; 

                    \draw[->,color=black!100] (0,-1.3) -- coordinate (y axis mid) (0,1.3); 
                    \draw[-,color=black!40] (0,1) -- coordinate (x axis mid) (5.2,1); 
                    \draw[-,color=black!40] (0,-1) -- coordinate (x axis mid) (5.2,-1); 

                    \draw[color=black!0] (3.5,0) -- (3.5,0) node[anchor=west,color=black] {$g$}; 

                    \draw[color=black,domain=0:4.71,samples=100]   plot (\x,{-sin(\x r)})   node[right]{};
                    \candle{4.95}{0}{1}{1}{-1};
            \end{tikzpicture}
            \label{figSinB}
          }
          \caption{IPFs from Example~\ref{bsp2} with the resulting candles.}
          \label{figSin}
	\end{figure}
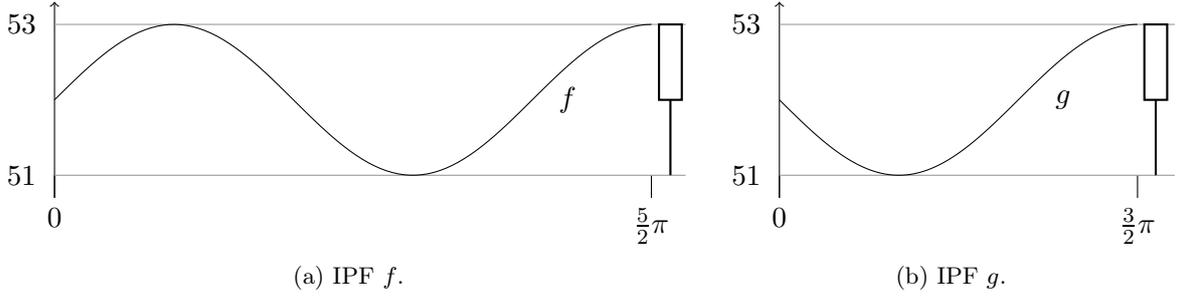

	Since we cannot check a backtest engine for each candle contained in the infinitely dimensional space $(\IR^+)^4$ we need to reduce the problem to finitely many model candles. However, for this step we need the backtest engine to be stable under transformation of candles which we discuss in the next section.


	\section{Suitable test cases for a proof of correctness}
	\label{sec:correctness}

	In this section we want to examine whether or not a given backtest engine works correctly.
	For this goal, we design ``model candles'', i.e. a set of finitely many candles which allow for a proof of correctness under the assumption of ``stability under transformations'', a concept which will be introduced as well.

	In general, the values of a CR do not coincide with the levels of the underlying setup, but we observe that possible values for entry and exit prices are the open of the candle and the levels of the setup.
	The exact values of open, close, high and low are not important for the resulting entry and exit, but rather their position relative to the levels $L_1,\ldots,L_m$.
	We therefore define $l_{2i-1}:=L_{i}$ for $i=1,\ldots,m$ and introduce intermediate levels $l_{2i}$ with $L_i<l_{2i}<L_{i+1}$ such that $l_0<l_1<\ldots<l_{2m}$.
	Restricting the candle data to these values results in a system of finitely many \emph{representative candles}.
	We observe that the number of representative candles is at most $(2m+1)^4$, as open, close, high and low can only take the $2m+1$ values $l_0,\ldots,l_{2m}$.

	\begin{sample}\label{bsp3}
		In Example~\ref{bsp1}, we may choose
		\begin{align}\label{eqLevels}
			l_0:=50<l_1=L_1=51<l_2:=52<l_3=L_2=53<l_4:=54.
		\end{align}
		Figure~\ref{repCandles} shows all representative candles for this setup.
	\end{sample}
	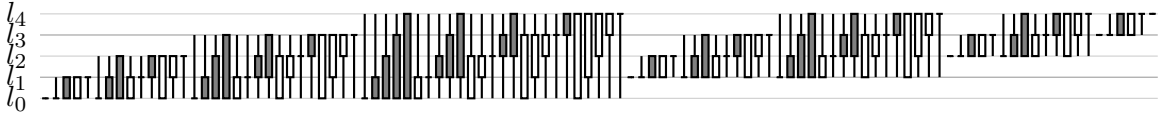
\begin{figure}
		\centering
		\begin{tikzpicture}[scale=0.28]
			\foreach \x in {0,...,4}{
				\draw[color=black!0] (-3pt,\x) -- (-3pt,\x) node[anchor=east,color=black] {$l_\x$}; 
			}
			\foreach \x in {1,...,2}{
				\draw[-,color=black!40] (0,2*\x-1) -- coordinate (x axis mid) (52.5,2*\x-1); 
			}
			\foreach \x in {0,...,2}{
				\draw[-,color=black!20] (0,2*\x) -- coordinate (x axis mid) (52.5,2*\x); 
			}
			\newcounter{x};
			\foreach \l in {0,...,4}{
			\foreach \h in {\l,...,4}{
			\foreach \c in {\l,...,\h}{
			\foreach \o in {\l,...,\h}{
				\candle{0.25+0.5*\value{x}}{\o}{\c}{\h}{\l};
				\addtocounter{x}{1};
			}
			}
			}
			}
		\end{tikzpicture}
		\caption{All 105 possible candles with values in $\set{l_0,\ldots,l_4}$ from Example~\ref{bsp3}.\label{repCandles}}
	\end{figure}

	In order to further restrict the levels $l_i$ to fixed values, we need to introduce a notion of transformation from one set of levels to another one.
	\begin{definition}
		A strictly monotonously increasing bijective function
		\[
			T: \IR^+\to\IR^+
		\]
		is in the following called a \emph{transformation}.
		In order to apply a transformation to results, we set $T(-1):=-1$ for entries and exits.
	\end{definition}

	A transformation can be used to transfer setups, candles, IPFs, results and CRs from one set of levels to another one.
	\begin{sample}\label{bsp4}
		Continuing Example~\ref{bsp3} with levels given in \eqref{eqLevels} and examining the corresponding CR $(52,51,53,51,53,51)$ of the IPF $f$ from Example~\ref{bsp2}, we might apply the transformation
		\[
			T: \IR^+\to\IR^+, t\mapsto 2t+1.
		\]
		This results in a setup with the same orders but at levels
		\[
			\tilde{l}_0:=101<\tilde{l}_1=\tilde{L}_1=103<\tilde{l}_2:=105<\tilde{l}_3=\tilde{L}_2=107<\tilde{l}_4:=109
		\]
		and a transformation of the CR given by
		$
			(105, 103, 107, 103, 107, 103)\in\IR^6.
		$

		Apparently, this tuple itself is a CR, i.e. it is derived from an IPF.
		This can be shown by using the IPF $f$ from Example~\ref{bsp2} and transforming it to $T\circ f$.
	\end{sample}

	Now we introduce the main assumption about the given backtest engine, which will allow us to prove correctness of that backtest engine.
	\begin{definition}\label{def:stableUT}
		We call a backtest engine $E$ \emph{stable under transformations}, if the following holds true for any given IPF $f$, transformation $T$ and setup at levels $L_1,\ldots,L_m$:\\
		If the result of the backtest engine $E$ with backtest mode $\mf{M}\in BM$ is given by
		\[
			E(C(f),\mf{M})=(entry,exit),
		\]
		then the result of the backtest engine after transformation for the new setup at levels $T(L_1),\ldots,$ $T(L_m)$ and the new IPF given by $T\circ f$ yields
		\[
			E(C(T\circ f),\mf{M})=(T(entry),T(exit))=:T(entry,exit).
		\]
	\end{definition}

	\begin{remark}\label{CRStable}
		For a backtest engine, it is desirable to be stable under transformations.
		This is due to the structure of orders which are independent of the exact levels but depend only on their relative values.
		The decision trees in \cite{MP2014}, which outline how correct results for several order setups can be obtained, are also independent of the exact levels and only depend on relative values.
		Therefore the decision tree and thus also the corresponding results are naturally stable under transformations.
		We conclude that results only depend on relative positions of orders, or, formally, for all IPFs $f$ we have $R(T\circ f)=T(R(f))$.

		Additionally, by the monotonicity of transformations, best cases remain best cases and worst cases remain worst cases under transformations, which can again be seen in the decision trees in \cite{MP2014}.
	\end{remark}

	From representative candles with values in $\set{l_0,\ldots,l_{2m}}$, we can now cover almost all possible candles by applying transformations, but some important cases are missing:
	The ``generic''~candle has a range $[low,high]$ that does not include any level from $\set{L_1,\ldots,L_m}$ but nonetheless the values for open, close, high and low may differ from each other (see Figure~\ref{figBlowUpA} for examples of such candles).
	These cases are not covered by the considerations made so far, but can easily be taken into account by the following considerations.
	Combining our observations, we arrive at the set of ``model'' candles we want to examine by introducing four intermediate levels between $L_{i-1}$ and $L_i$.

	\begin{definition}\label{modelBar}
		Let a setup at levels $L_1,\ldots,L_m$ be given.
		Furthermore, choose values $l_{2i},l_{2i-1},$ $l_{2i,j}\in\IR^+$ for $i=1,\ldots,m$ and $j=1,\ldots,4$ with
		\begin{gather*}
			l_{0}=l_{0,1}<l_{0,2}<l_{0,3}<l_{0,4}<L_0,\\
			l_{2i-1} := L_i<l_{2i}=l_{2i,1}<l_{2i,2}<l_{2i,3}<l_{2i,4}<l_{2i+1}:=L_{i+1}\quad\text{ for }1\leq i\leq m-1,\\
			L_m<l_{2m}=l_{2m,1}<l_{2m,2}<l_{2m,3}<l_{2m,4}.
		\end{gather*}
		A \emph{model candle} is a candle with values
		\[
			open,close,high,low\in \set{l_i\mid 0\leq i\leq 2m}\cup\set{l_{2i,j}\mid 0\leq i\leq m,1\leq j\leq 4}
		\]
		with the further restriction
		\begin{multline}\label{eq:sublevel_condition}
			\forall i\in\set{0,\ldots,m}\exists r_i\in\set{1,\ldots,4}:
			\text{ if } \set{open,close,high,low}\cap\set{l_{2i,j}\mid1\leq j\leq 4}\neq\emptyset \\
			\text{then }\set{open,close,high,low}\cap\set{l_{2i,j}\mid1\leq j\leq 4}=\set{l_{2i,j}\mid1\leq j\leq r_i}.
		\end{multline}
	\end{definition}

	Figure~\ref{figBlowUp} shows examples of some candles, where in this extremal case we have up to $12$ model candles for one of the representative candles introduced.

	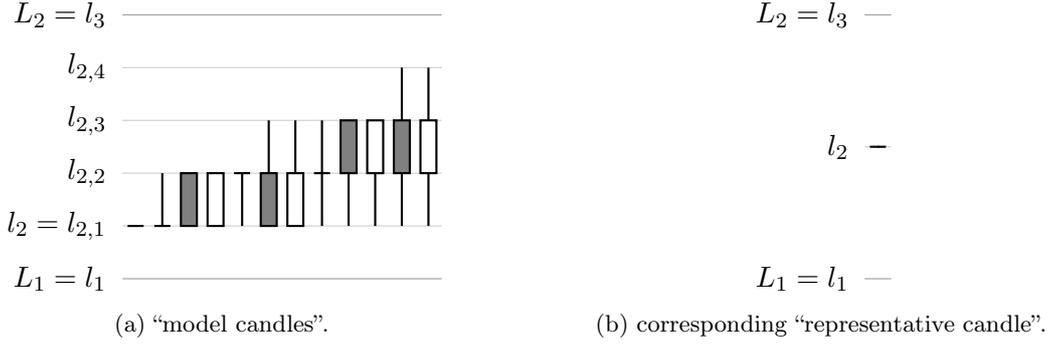
\begin{figure}
	\centering
		\subfloat[``model candles''.]{
		\begin{minipage}{7cm}
		\centering
		\begin{tikzpicture}[scale=0.7]
			\draw[-,color=black!40] (0,0) -- coordinate (x axis mid) (6,0); 
			\draw[color=black!0] (-3pt,0) -- (-3pt,0) node[anchor=east,color=black] {$L_1=l_1$}; 

			\draw[-,color=black!20] (0,1) -- coordinate (x axis mid) (6,1); 
			\draw[color=black!0] (-3pt,1) -- (-3pt,1) node[anchor=east,color=black] {$l_2=l_{2,1}$}; 
			\draw[-,color=black!20] (0,2) -- coordinate (x axis mid) (6,2); 
			\draw[color=black!0] (-3pt,2) -- (-3pt,2) node[anchor=east,color=black] {$l_{2,2}$}; 
			\draw[-,color=black!20] (0,3) -- coordinate (x axis mid) (6,3); 
			\draw[color=black!0] (-3pt,3) -- (-3pt,3) node[anchor=east,color=black] {$l_{2,3}$}; 
			\draw[-,color=black!20] (0,4) -- coordinate (x axis mid) (6,4); 
			\draw[color=black!0] (-3pt,4) -- (-3pt,4) node[anchor=east,color=black] {$l_{2,4}$}; 

			\draw[-,color=black!40] (0,5) -- coordinate (x axis mid) (6,5); 
			\draw[color=black!0] (-3pt,5) -- (-3pt,5) node[anchor=east,color=black] {$L_2=l_3$}; 

			\candle{0.25}{1}{1}{1}{1};
			\candle{0.75}{1}{1}{2}{1};
			\candle{1.25}{2}{1}{2}{1};
			\candle{1.75}{1}{2}{2}{1};
			\candle{2.25}{2}{2}{2}{1};
			\candle{2.75}{2}{1}{3}{1};
			\candle{3.25}{1}{2}{3}{1};
			\candle{3.75}{2}{2}{3}{1};
			\candle{4.25}{3}{2}{3}{1};
			\candle{4.75}{2}{3}{3}{1};
			\candle{5.25}{3}{2}{4}{1};
			\candle{5.75}{2}{3}{4}{1};
		\end{tikzpicture}
		\end{minipage}
		\label{figBlowUpA}
		}
		\hspace{5mm}
		\subfloat[corresponding ``representative candle''.]{
		\begin{minipage}{7cm}
		\centering
		\begin{tikzpicture}[scale=0.7]
			\draw[-,color=black!40] (0,0) -- coordinate (x axis mid) (0.5,0); 
			\draw[color=black!0] (-3pt,0) -- (-3pt,0) node[anchor=east,color=black] {$L_1=l_1$}; 

			\draw[-,color=black!20] (0,2.5) -- coordinate (x axis mid) (0.5,2.5); 
			\draw[color=black!0] (-3pt,2.5) -- (-3pt,2.5) node[anchor=east,color=black] {$l_2$}; 

			\draw[-,color=black!40] (0,5) -- coordinate (x axis mid) (0.5,5); 
			\draw[color=black!0] (-3pt,5) -- (-3pt,5) node[anchor=east,color=black] {$L_2=l_3$}; 

			\candle{0.25}{2.5}{2.5}{2.5}{2.5};
		\end{tikzpicture}
		\end{minipage}
		\label{figBlowUpB}
		}

		\caption{An example of candles with values in $\set{l_{2,j}\mid j=1,\ldots,4}$ (see Definition~\ref{modelBar}).}
		\label{figBlowUp}
	\end{figure}

	In this way, we obtain sufficiently many candles to cover the ``generic''~candles described above as well.

	\begin{remarks}
		Condition \eqref{eq:sublevel_condition} for the sublevels $l_{2i,j}$ is introduced to minimize the use of these additional levels by, figuratively speaking, filling up the levels $l_{2i,j}$ for a fixed $i$ by increasing $j$ as far as necessary.
		This reduces the total number of model candles.

		For an application of the construction, all values can be chosen equidistant.
		It is important to consider the tick size used and especially to not choose the sublevels $l_{2i,j}$ with too little distance.
		For example if we have a tick size of $0.01$ we could choose $l_i:=50.05+i$ and $l_{2i,j}:=49.95+2i+0.1j$ for $i=0,\ldots,m,j=1,\ldots,4$, which allows us to also check for rounding errors.
	\end{remarks}

	With all these preparations, we are now able to conduct a proof of correctness for given backtest engine and setup.
	\begin{thm}\label{thm:correct_engine}
		Let a backtest engine $E$ be given which is stable under transformations for a setup at (ordered) levels $L_1<\ldots<L_m$.
		If the backtest engine works correctly (e.g. for the criteria given in \cite{MP2014}) on the set of all model candles of this setup with fixed $l_i,l_{2i,j}$ as in Definition~\ref{modelBar}, then it works correctly on this setup with any (ordered) levels $\tilde{L}_1<\ldots<\tilde{L}_m$ (and arbitrary candles).
	\end{thm}
	\begin{proof}
		Let an arbitrary candle $c=(open,close,high,low)\in(\IR^+)^4$, ordered levels $\tilde{L}_1,\ldots,\tilde{L}_m\in\IR^+$ and a backtest mode $\mf{M}$ be given.
		Set $(\widetilde{entry},\widetilde{exit}):=E(c,\mf{M})$.

		We have to show that $(\widetilde{entry},\widetilde{exit})$ is a correct result for the candle $c$, i.e. that there exists some IPF $\tilde{f}$ with $c=C(\tilde{f})$ and $R(\tilde{f})=(\widetilde{entry},\widetilde{exit})$, or, if $\mf{M}=ignore$ and there is no unique result, that $(\widetilde{entry},\widetilde{exit})=(-1,-1)$.

		In order to do this let us define $\tilde{l}_{2i-1}:=\tilde{L}_i$ for $1\leq i\leq m$ and $\tilde{l}_{2i,j},\tilde{l}_i\in\IR^+$ such that the requirements from Definition~\ref{modelBar} are fulfilled and $open,close,high,low\in\set{\tilde{l}_i\mid 0\leq i\leq 2m}\cup\set{\tilde{l}_{2i,j}\mid 0\leq i\leq m,1\leq j\leq 4}$ which fulfill the minimal sublevel condition~\eqref{eq:sublevel_condition} for the level $\tilde{l}_{2i,j}$ instead of $l_{2i,j}$.
		Such a choice is obviously always possible.

		Now we construct a transformation $T$ that transforms the setup at levels $\tilde{L}_i$ to the setup at levels $L_i$.
		This can be done by setting $T(0)=0$, $T(\tilde{l}_{2i-1})=l_{2i-1}$, $T(\tilde{l}_{2i,j})=l_{2i,j}$ for $i=0,\ldots, m$ and $j=1,\ldots,4$ and interpolating piecewise linearly in between these values.
		Formally,
		\[
			T:\IR^+\to\IR^+, t\mapsto
			\begin{cases}
				\frac{l_{0,1}}{\tilde{l}_{0,1}}t,&t<\tilde{l}_{0,1},\\
				l_{2i,j}+(t-\tilde{l}_{2i,j})\frac{l_{2i,j+1}-l_{2i,j}}{\tilde{l}_{2i,j+1}-\tilde{l}_{2i,j}},&\tilde{l}_{2i,j}\leq t<\tilde{l}_{2i,j+1},0\leq i\leq m,1\leq j\leq 3,\\
				l_i+(t-\tilde{l}_i)\frac{l_{i+1,1}-l_i}{\tilde{l}_{i+1,1}-\tilde{l}_i},&\tilde{l}_{i}\leq t<\tilde{l}_{i+1,1},i\text{ odd},\\
				l_{i-1,4}+(t-\tilde{l}_{i-1,4})\frac{l_i-l_{i-1,4}}{\tilde{l}_i-\tilde{l}_{i-1,4}},&\tilde{l}_{i-1,4}\leq t<\tilde{l}_{i},i\text{ odd},\\
				l_{2m,4}-\tilde{l}_{2m,4}+t,&t\geq\tilde{l}_{2m,4}.
			\end{cases}
		\]

		By construction, $T$ is continuous, strictly monotonously increasing and bijective.
		Furthermore, $T(c):=(T(open),T(close),T(high),T(low))$ is a model candle w.r.t. $L_i,l_i,l_{i,j}$ (cf. Figure~\ref{transformModelCandle}).

		\begin{figure}
		\centering
                \subfloat[Candle $c$ and levels $\tilde{L}_i,\tilde{l}_i,\tilde{l}_{i,j}$.]{
		\begin{minipage}{6cm}
		\centering
		\begin{tikzpicture}[scale=0.7]
			\draw[-,color=black!40] (0,0) -- coordinate (x axis mid) (1,0); 
			\draw[color=black!0] (-3pt,0) -- (-3pt,0) node[anchor=east,color=black] {$\tilde{L}_i=\tilde{l}_{2i-1}$}; 

			\draw[-,color=black!20] (0,1.5) -- coordinate (x axis mid) (1,1.5); 
			\draw[color=black!0] (-3pt,1.5) -- (-3pt,1.5) node[anchor=east,color=black] {$\tilde{l}_{2i}=\tilde{l}_{2i,1}$}; 
			\draw[-,color=black!20] (0,2) -- coordinate (x axis mid) (1,2); 
			\draw[color=black!0] (1,2) -- (1,2) node[anchor=west,color=black] {$\tilde{l}_{2i,2}$}; 
			\draw[-,color=black!20] (0,4) -- coordinate (x axis mid) (1,4); 
			\draw[color=black!0] (-3pt,4) -- (-3pt,4) node[anchor=east,color=black] {$\tilde{l}_{2i,3}$}; 
			\draw[-,color=black!20] (0,4.8) -- coordinate (x axis mid) (1,4.8); 
			\draw[color=black!0] (1,4.8) -- (1,4.8) node[anchor=west,color=black] {$\tilde{l}_{2i,4}$}; 

			\draw[-,color=black!40] (0,5) -- coordinate (x axis mid) (1,5); 
			\draw[color=black!0] (-3pt,5) -- (-3pt,5) node[anchor=east,color=black] {$\tilde{L}_{i+1}=\tilde{l}_{2i+1}$}; 

			\candle{0.5}{1.5}{2}{4}{0};
		\end{tikzpicture}
		\end{minipage}%
		\label{transformModelCandleA}
		}
		\hspace{5mm}
                \subfloat[Candle $T(c)$, a model candle w.r.t. $L_i,l_i,l_{i,j}$.]{
		\begin{minipage}{8cm}
		\centering
		\begin{tikzpicture}[scale=0.7]
			\draw[-,color=black!40] (0,0) -- coordinate (x axis mid) (1,0); 
			\draw[color=black!0] (-3pt,0) -- (-3pt,0) node[anchor=east,color=black] {$L_i=l_{2i-1}$}; 

			\draw[-,color=black!20] (0,1) -- coordinate (x axis mid) (1,1); 
			\draw[color=black!0] (-3pt,1) -- (-3pt,1) node[anchor=east,color=black] {$l_{2i}=l_{2i,1}$}; 
			\draw[-,color=black!20] (0,2) -- coordinate (x axis mid) (1,2); 
			\draw[color=black!0] (-3pt,2) -- (-3pt,2) node[anchor=east,color=black] {$l_{2i,2}$}; 
			\draw[-,color=black!20] (0,3) -- coordinate (x axis mid) (1,3); 
			\draw[color=black!0] (-3pt,3) -- (-3pt,3) node[anchor=east,color=black] {$l_{2i,3}$}; 
			\draw[-,color=black!20] (0,4) -- coordinate (x axis mid) (1,4); 
			\draw[color=black!0] (-3pt,4) -- (-3pt,4) node[anchor=east,color=black] {$l_{2i,4}$}; 

			\draw[-,color=black!40] (0,5) -- coordinate (x axis mid) (1,5); 
			\draw[color=black!0] (-3pt,5) -- (-3pt,5) node[anchor=east,color=black] {$L_{i+1}=l_{2i+1}$}; 

			\candle{0.5}{1}{2}{3}{0};
		\end{tikzpicture}
		\end{minipage}
		\label{transformModelCandleB}
		}

		\caption{An example of the application of the transformation $T$ on an arbitrary candle.}
		\label{transformModelCandle}
		\end{figure}
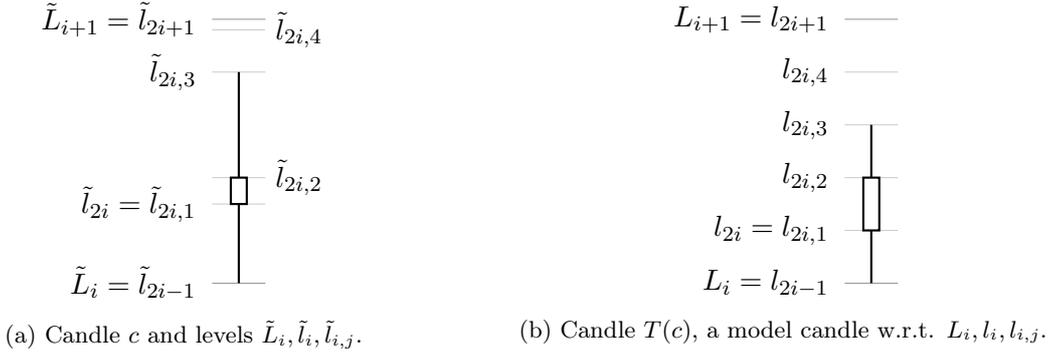

		By assumption, $E$ works correctly on model candles w.r.t. $L_i,l_i,l_{i,j}$. Therefore, we know that $E(T(c),\mf{M})$ is the desired result of the model candle $T(c)$ and $(T(c),E(T(c),\mf{M}))$ the desired CR.

		If $\mf{M}=ignore$ and there is no unique result for $T(c)$, this still holds true after applying the inverse transformation $T\inv$, so we have $(-1,-1)=E(T(c),\mf{M})=E(c,\mf{M})$ by stability under transformations, which is the correct result by Remark \ref{CRStable}.

		Otherwise, i.e. $\mf{M}\neq ignore$ or the result for $T(c)$ is unique, there exists an IPF $f=f_{T(c),\mf{M}}$ (depending on $T(c)$ and $\mf{M}$ only) such that $C(f)=T(c)$ and $R(f)=E(T(c),\mf{M})$, because each result corresponds to some IPF.
		Then clearly $C(T\inv\circ f)=c$.
		Therefore setting $\tilde{f}:=T\inv\circ f$, we already have $C(\tilde{f})=c$ as desired.
		Furthermore, applying the stability of $E$ under the transformation $T\inv$ (cf. Definition \ref{def:stableUT}), we obtain
		\[
			(\widetilde{entry},\widetilde{exit})\overset{\text{def.}}{=}E(c,\mf{M})=E(C(T\inv\circ f),\mf{M})\overset{\text{stability u.t. of }E}{=}T\inv(E(C(f),\mf{M})).
		\]
		Since, by construction, $C(f)=T(c)$ and $R(f)=E(T(c),\mf{M})$, we conclude from the stability under transformation of results (cf. Remark \ref{CRStable}) that
		\[
			(\widetilde{entry},\widetilde{exit})=T\inv(E(T(c),\mf{M}))=T\inv(R(f))=R(\tilde{f}),
		\]
		which is, again by Remark \ref{CRStable}, the correct result.

		As $c$ and $\tilde{L}_1,\ldots,\tilde{L}_m$ were chosen arbitrarily, we conclude general correctness of $E$ for this setup.
	\end{proof}

	Theorem~\ref{thm:correct_engine} allows us to reduce the verification of a stable backtest engine to simply check the backtest evaluation on one set of model candles. Therefore this set is focused in the next section.


	\section{A model for intra-period prices}
	\label{sec:IPFs}

	Our goal in this section is to develop a model for IPFs which allows for only finitely many model price functions but still results in all possible CRs.
	The results for all CRs are needed to test a given backtest engine by comparing its results with correct results.
	They can be determined by the means of this section with the help of a correct reference backtest engine on the finitely many model price functions.

	We can restrict ourselves to model candles which do not take the values $l_{2i,2},l_{2i,3},l_{2i,4}$, but only $l_{2i}=l_{2i,1}$ for all $1\leq i\leq m$, because, as we see below, we can reconstruct the results for the model candles with values in $l_{2i,2},l_{2i,3},l_{2i,4}$.
	\begin{remark}\label{restrictRepCandle}
		Let a setup and a backtest mode $\mf{M}$ be given.

		Furthermore, let $c=(open,close,high,low)$ be a model candle w.r.t. the levels
		\[
			\mc{L}:=\set{l_i\mid 0\leq i\leq 2m}\cup\set{l_{2i,j}\mid 0\leq i\leq m,1\leq j\leq 4},
		\]
		i.e. $c\in \mc{L}^4$.
		We can obtain the result for $c$ by first applying the substitution
		\[
			l_{2i,j}\mapsto l_{2i}\text{ for all }0\leq i\leq m,\,1\leq j\leq 4
		\]
		to $c$ which leads to a new candle $c'$. Calculating the result $r'$ of $c'$ and then ``re-substituting''
		\[
			l_{2i_0}\mapsto l_{2i_0,j_0}\text{ in case }l_{2i_0,j_0}=open\text{ of the candle }c
		\]
		in $r'$ gives us the result $r$.
		Indeed, $r$ is the desired result for $c$.

		This procedure is possible due to the fact that entry and exit can only occur at levels $L_i$ or at the open value.
	\end{remark}

	\begin{sample}
		Remark~\ref{restrictRepCandle} can be illustrated for our setup from Example~\ref{bsp1} (EnterLongStop at $L_2$, stop loss at $L_1$) with the model candle
		\[
			c:=(open=l_{4,2},close=l_{4,1}=l_4,high=l_{4,3},low=l_3=L_2),
		\]
		which has the unique result $r:=(entry=open=l_{4,2},exit=-1)$.
		By substituting $l_{4,j}$ with $l_4$ for all $1\leq j\leq 4$, we get the representative candle
		\[
			c':=(open'=l_{4},close'=l_4,high'=l_{4},low'=l_3=L_2)
		\]
		with the unique result $r':=(entry'=open'=l_{4},exit'=-1)$.
		We can reconstruct $r$ from $r'$ by substituting $l_4=open'\mapsto open=l_{4,2}$ in $r'$.
	\end{sample}

	The following model for intra period prices allows only for the values $l_0,\ldots,l_{2m}$ at certain interpolation points.

	\begin{definition}
		An \emph{intra-period model price series} (IPMS) w.r.t. a given setup is a finite sequence $s:=(l_{i_1},\ldots,l_{i_k})$ with $0\leq i_j\leq 2m$ for all $1\leq j\leq k$ and $i_{j+1}=i_j+1$ or $i_{j+1}=i_j-1$ for all $1\leq j\leq k-1$.
		We call $|s|:=k$ the \emph{size} of the IPMS.
		By $\mc{M}:=\mc{M}_m$, we denote the set of all IPMS w.r.t. a number of levels $m$.
	\end{definition}

	The connection with intra-period price functions is the following:

	\begin{observation}
		Every IPMS $s=(l_{i_1},\ldots,l_{i_k})$ can be interpreted as an IPF by using piecewise linear interpolation as follows:
		Define $f^{(s)}:[1,k]\to\IR^+$ by
		\[
			f^{(s)}(t):=l_{i_{\lfloor t\rfloor}}+(t-\lfloor t\rfloor)(l_{i_{\lfloor t\rfloor+1}}-l_{i_{\lfloor t\rfloor}}),
			\quad\text{where }l_{i_{k+1}}:=l_{i_k},
		\]
		such that $f^{(s)}(j)=l_{i_j}$ for $j=1,\ldots,k$.
		Indeed, $f^{(s)}$ is continuous by construction.
	\end{observation}

	Through this connection, we define the result of an IPMS.

	\begin{definition}
		For a given setup with $m$ levels, we define
		\[
			\mc{R}:\mc{M}\to(\IR^+\cup\set{-1})^2, s\mapsto R(f^{(s)})
		\]
		and call $\mc{R}(s)$ the result of $s$.

		Similarly, by $\mc{C}(s):=C(f^{(s)})$ we denote the candle resulting from $s$.
	\end{definition}

	\begin{sample}\label{bsp5}
		In Example~\ref{bsp1}, an IPMS could be given by $s:=(l_2=52,l_3=53,l_2=52,l_1=51,l_2=52,l_3=53)$.
		The corresponding piecewise linear IPF $f^{(s)}$ in a graphical representation is given in Figure~\ref{figIPMS}.
	\end{sample}
	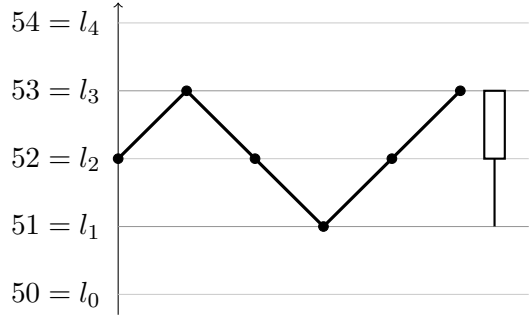
\begin{figure}
	\centering
	\begin{tikzpicture}[scale=0.9]
		\draw[->,color=black!100] (0,-0.3) -- coordinate (y axis mid) (0,4.3); 

		\foreach \x in {0,...,4}{
			\draw[color=black!0] (-3pt,\x) -- (-3pt,\x) node[anchor=east,color=black] {$5\x=l_\x$}; 
		}

		\foreach \x in {1,...,2}{
			\draw[-,color=black!40] (0,2*\x-1) -- coordinate (x axis mid) (6,2*\x-1); 
		}
		\foreach \x in {0,...,2}{
			\draw[-,color=black!20] (0,2*\x) -- coordinate (x axis mid) (6,2*\x); 
		}
		\draw [very thick, color=black!100] (0,2) -- (1,3) -- (2,2) -- (3,1) -- (4,2) -- (5,3);
		\draw[black, fill=black] (0,2) circle[radius=2pt]; 
		\draw[black, fill=black] (1,3) circle[radius=2pt];
		\draw[black, fill=black] (2,2) circle[radius=2pt];
		\draw[black, fill=black] (3,1) circle[radius=2pt];
		\draw[black, fill=black] (4,2) circle[radius=2pt];
		\draw[black, fill=black] (5,3) circle[radius=2pt];

		\candle{5.5}{2}{3}{3}{1};
	\end{tikzpicture}
	\caption{IPF $f^{(s)}$ for IPMS $s$ in Example~\ref{bsp5} with the resulting candle.}
	\label{figIPMS}
	\end{figure}

	This example has the same CR as the IPF $f$ in Example~\ref{bsp2}, which is generalized in the following theorem.

	\begin{thm}\label{contModel}
		Every CR $r\in \set{l_0,\ldots,l_{2m}}^4\times(\IR^+\cup\set{-1})^2$ is the CR of an IPMS.
	\end{thm}
	\begin{proof}
		Let $f\in\mc{C}([a,b],\IR^+)$ be an IPF with $C(f)\in\set{l_0,\ldots,l_{2m}}^4$ such that $(C(f),R(f))=r$ for an arbitrary but fixed CR $r$.
		Such an $f$ exists by the definition of CRs, see Definition~\ref{def:CR}.

		In order to construct an IPMS with the same result $r$, we use piecewise linear interpolation:
		W.l.o.g. we assume that $f$ assumes values in $\set{l_0,\ldots,l_{2m}}$ only at finitely many discrete points, because otherwise we can easily modify $f$ in such a way without changing the CR $r$.
		Now define $t_1,\ldots,t_k$ s.t.
		\[\begin{split}
			&t_j<t_{j+1}\text{ for all }1\leq j\leq k-1,\\
			&f(t_j)\in\set{l_0,\ldots,l_{2m}}\text{ for all }1\leq j\leq k,\text{ and}\\
			&f(t)\notin\set{l_0,\ldots,l_{2m}}\text{ for all } t\in[a,b]\setminus\set{t_1,\ldots,t_k}.
		\end{split}\]

		Next, we remove all those $t_j$ from the $\set{t_1,\ldots,t_k}$ with $f(t_j)=f(t_{j-1})$ (remember that one requirement for an IPMS $s$ is that $s_{j+1}\neq s_j$) by defining $t'_j$ s.t.
		\[
			\set{t'_1,\ldots,t'_{k'}}=\set{t_1,\ldots,t_k}\setminus\set{t_j\mid f(t_j)=f(t_{j-1})}\quad\text{and}\quad t'_j<t'_{j+1}\text{ for all }1\leq j\leq k'-1.
		\]
		By continuity of $f$ and the intermediate value theorem, $s:=(f(t'_1),\ldots,f(t'_{k'}))$ is an IPMS and, again by the intermediate value theorem, $\mc{R}(s)=R(f)=r$.
	\end{proof}
	Example~\ref{bsp5} gives actually the IPMS obtained from $f$ in Example~\ref{bsp2} by the construction in the proof of Theorem~\ref{contModel}.

	For an algorithmic approach, we want to limit ourselves to finitely many IPMSs from which all CRs can be obtained.
	For this, we need to limit the size of the IPMSs.

	\begin{definition}
		By $M_n$ we denote the set of all CRs of IPMSs of size of at most $n$, i.e.
		\[
			M_n:=\set{r\mid \exists s\in\mc{M},\abs{s}\leq n,r=(\mc{C}(s),\mc{R}(s))}.
		\]
	\end{definition}

	Now we can limit the size of IPMSs used to obtain all possible CRs for a given setup.

	\begin{thm}\label{upperLimit}
		For given setup and $n_0\in\IN$, the following holds true:
		If $M_{n_0}=M_{n_0+1}$, then $M_n=M_{n_0}$ for all $n\geq n_0$.
	\end{thm}
	\begin{proof}
		We show that $M_n=M_{n+1}$ implies $M_{n+1}=M_{n+2}$.
		Then, the claim follows inductively.

		For the following construction we perform four steps. An example for these steps are illustrated in four images in Figure~\ref{figSteps}, respectively, to which we refer correspondingly.

		Let $M_{n}=M_{n+1}$ and $s=(l_{i_1},\ldots,l_{i_{n+2}})$ be an IPMS of size $n+2$ with CR $r=(\mc{C}(s),\mc{R}(s))=(open,close,high,low,entry,exit)\in M_{n+2}$ (cf. Figure~\ref{figStepsA}).
		We have to construct an IPMS of size of at most $n+1$ which has the same CR $r$.

		Consider the IPMS $s'=(l_{i_1},\ldots,l_{i_{n+1}})$ which is the same IPMS as $s$ but shortened by the last element, i.e. we have $s=(s',l_{i_{n+2}})$ (cf. Figure~\ref{figStepsB}). The size of $s'$ is $n+1$ and we denote its CR by $r'=(\mc{C}(s'),\mc{R}(s'))=(open',close',high',low',entry',exit')\in M_{n+1}$.
		By definition of an IPMS we have $\abs{close'-l_{i_{n+2}}}=1$.
		By assumption we have $r'\in M_n$, i.e. there exists an IPMS $\bar{s}'$ of size $\abs{\bar{s}'}\leq n$ with CR $r'$ (cf. Figure~\ref{figStepsC}).
		W.l.o.g. let $\abs{\bar{s}'}=n$ (otherwise, the same proof can be conducted with a different index in use).

		We construct $\bar{s}\in\IR^{n+1}$ as the concatenation $\bar{s}:=(\bar{s}',l_{i_{n+2}})$ with CR $\bar{r}=(\mc{C}(\bar{s}),\mc{R}(\bar{s}))=(\overline{open},\overline{close},\overline{high},\overline{low},\overline{entry},\overline{exit})$ (cf. Figure~\ref{figStepsD}).
		Indeed, $\bar{s}$ is an IPMS, as $\bar{s}_n=\bar{s}'_n=close'$ and by definition $\abs{\bar{s}_n-\bar{s}_{n+1}}=\abs{close'-l_{i_{n+2}}}=1$.

		We now claim that $\bar{s}$ has the CR $r$, i.e. $\bar{r}=r$.
		By the structure as concatenation it is obvious that $\overline{open}=open'=open$ and $\overline{close}=l_{i_{n+2}}=close$.
		Furthermore, the high and low values can be determined as maximum and minimum, respectively:
		\[\begin{split}
			\overline{high}&=\max\set{\max{\bar{s}'},l_{i_{n+2}}}\overset{\mc{C}(\bar{s}')=\mc{C}(s')}{=}\max\{high',l_{i_{n+2}}\}=\max\{l_{i_1},\ldots,l_{i_{n+1}},l_{i_{n+2}}\}=high,\\
			\overline{low}&=\min\set{low',l_{i_{n+2}}}=low.
		\end{split}\]
		Therefore we have $\mc{C}(\bar{s})=\mc{C}(s)$.

		If there is no entry for $s$, the same holds true for $\bar{s}$ because the range $[low,high]=[\overline{low},\overline{high}]$ contains no level of an entry order.
		If the entry occurs during the first $n+1$ interpolation points of $s$, we get $\overline{entry}=entry'=entry$.
		Otherwise the entry of $s$ was executed on the $(n+2)$nd interpolation point with value $entry=l_{i_{n+2}}$.
		Then $entry'$ is $-1$ and we obtain $\overline{entry}=l_{i_{n+2}}$.
		Therefore $entry=\overline{entry}$.
		Analogously, it follows that $\overline{exit}=exit$.

		In conclusion, $\mc{R}(\bar{s})=\mc{R}(s)$ and thus $\bar{r}=r$, but $\abs{\bar{s}}=n+1$ and therefore $r\in M_{n+1}$.
	\end{proof}

	This allows for an algorithmic approach to find all relevant CRs.

	\begin{cor}\label{nSuff}
		In order to obtain all possible CRs of model candles with values in $\set{l_0,\ldots,l_{2m}}$ for a given setup, it suffices to compute $M_{n_0}$, where $n_0=\min\set{n\mid M_n=M_{n+1}}$.
	\end{cor}
	\begin{proof}
		The claim follows from Theorem~\ref{contModel} and Theorem~\ref{upperLimit}.
	\end{proof}

	\begin{remark}
		In the proof of Theorem~\ref{upperLimit}, we showed that for $\bar{s}'$ and $s'$ the equality
		\[
			(\mc{C}(\bar{s}'),\mc{R}(\bar{s}'))=(\mc{C}(s'),\mc{R}(s'))
		\]
		carries forward after appending $l_{i_{n+2}}$ as
		\[
			(\mc{C}(\bar{s}',l_{i_{n+2}}),\mc{R}(\bar{s}',l_{i_{n+2}}))=(\mc{C}(s',l_{i_{n+2}}),\mc{R}(s',l_{i_{n+2}})).
		\]
		Analogously and by induction, we obtain the following:
		If the CRs of two IPMSs $s_1$ and $s_2$ are the same, so are the CRs of all IPMSs obtained from $s_1$ and $s_2$ by concatenation of arbitrary consistent IPMSs, i.e. if $\mc{C}(s_1)=\mc{C}(s_2)$ and $\mc{R}(s_1)=\mc{R}(s_2)$ then
		\begin{gather*}
			\text{for all } k\in\IN, s_3\in\IR^k\text{ such that }(s_1,s_3)\text{ and }(s_2,s_3)\text{ are IPMSs, the following holds true:}\\
			\mc{C}((s_1,s_3))=\mc{C}((s_2,s_3))\text{ and }\mc{R}((s_1,s_3))=\mc{R}((s_2,s_3)).
		\end{gather*}
		Therefore an algorithm only needs to consider one of these IPMSs when appending other IPMSs, which yields an enormous speedup for the calculation of all model candles and their results.
	\end{remark}

	\begin{sample}\label{bspShorten}
		Continuing Example~\ref{bsp1}, we can apply the results.
		In order to abbreviate notation, we set $l_i:=i$ here.
		Algorithmically we found the value $n_0=11$ to be sufficiently large in the sense of Corollary~\ref{nSuff}.
		The example for the constructions from the proof of Theorem~\ref{upperLimit} shown in Figure~\ref{figSteps} has the following data:
		\[\begin{split}
			s=&(1,2,3,4,3,2,1,0,1,2,3,4,3),~\abs{s}=13,\\
			&(\mc{C}(s),\mc{R}(s))=(open=1,close=3,high=4,low=0,entry=3,exit=1),\\
			s'=&(1,2,3,4,3,2,1,0,1,2,3,4),~\abs{s'}=12,\\
			&(\mc{C}(s'),\mc{R}(s'))=(open'=1,close'=4,high'=4,low'=0,entry'=3,exit'=1),\\
			\bar{s}'=&(1,2,3,2,1,0,1,2,3,4),~\abs{\bar{s}'}=10,~(\mc{C}(\bar{s}'),\mc{R}(\bar{s}'))=(\mc{C}(s'),\mc{R}(s')),\\
			\bar{s}=&(1,2,3,2,1,0,1,2,3,4,3),~\abs{\bar{s}}=11,~(\mc{C}(\bar{s}),\mc{R}(\bar{s}))=(\mc{C}(s),\mc{R}(s)).
		\end{split}\]

		\begin{figure}[bth]
		\centering
		\subfloat[Initial IPMS $s$.]{
		\begin{tikzpicture}[scale=0.45]
			\draw[->,color=black!100] (0,-0.3) -- coordinate (y axis mid) (0,4.3); 

			\foreach \x in {0,...,4}{
				\draw[color=black!0] (-3pt,\x) -- (-3pt,\x) node[anchor=east,color=black] {$\x=l_\x$}; 
			}

			\foreach \x in {1,...,2}{
				\draw[-,color=black!40] (0,2*\x-1) -- coordinate (x axis mid) (13,2*\x-1); 
			}
			\foreach \x in {0,...,2}{
				\draw[-,color=black!20] (0,2*\x) -- coordinate (x axis mid) (13,2*\x); 
			}

			\draw [very thick, color=black!100] (0,1) -- (1,2) -- (2,3) -- (3,4) -- (4,3) -- (5,2) -- (6,1) -- (7,0) -- (8,1) -- (9,2) -- (10,3) -- (11,4) -- (12,3);

			\candle{12.5}{1}{3}{4}{0};
		\end{tikzpicture}
		\label{figStepsA}
		}
		\hfill
		\subfloat[Shortened IPMS $s'$.]{
		\begin{tikzpicture}[scale=0.45]
			\draw[->,color=black!100] (0,-0.3) -- coordinate (y axis mid) (0,4.3); 

			\foreach \x in {0,...,4}{
				\draw[color=black!0] (-3pt,\x) -- (-3pt,\x) node[anchor=east,color=black] {$\x=l_\x$}; 
			}

			\foreach \x in {1,...,2}{
				\draw[-,color=black!40] (0,2*\x-1) -- coordinate (x axis mid) (13,2*\x-1); 
			}
			\foreach \x in {0,...,2}{
				\draw[-,color=black!20] (0,2*\x) -- coordinate (x axis mid) (13,2*\x); 
			}

			\draw [very thick, color=black!100] (0,1) -- (1,2) -- (2,3) -- (3,4) -- (4,3) -- (5,2) -- (6,1) -- (7,0) -- (8,1) -- (9,2) -- (10,3) -- (11,4);

			\candle{12.5}{1}{4}{4}{0};
		\end{tikzpicture}
		\label{figStepsB}
		}

		\subfloat[IPMS $\bar{s}'$ with the same result as $s'$, but smaller size.]{
		\begin{tikzpicture}[scale=0.45]
			\draw[->,color=black!100] (0,-0.3) -- coordinate (y axis mid) (0,4.3); 

			\foreach \x in {0,...,4}{
				\draw[color=black!0] (-3pt,\x) -- (-3pt,\x) node[anchor=east,color=black] {$\x=l_\x$}; 
			}

			\foreach \x in {1,...,2}{
				\draw[-,color=black!40] (0,2*\x-1) -- coordinate (x axis mid) (13,2*\x-1); 
			}
			\foreach \x in {0,...,2}{
				\draw[-,color=black!20] (0,2*\x) -- coordinate (x axis mid) (13,2*\x); 
			}

			\draw [very thick, color=black!100] (0,1) -- (1,2) -- (2,3) -- (3,2) -- (4,1) -- (5,0) -- (6,1) -- (7,2) -- (8,3) -- (9,4);

			\candle{12.5}{1}{4}{4}{0};
		\end{tikzpicture}
		\label{figStepsC}
		}
		\hfill
		\subfloat[IPMS $\bar{s}$ with the same result as $s$, but smaller size.]{
		\begin{tikzpicture}[scale=0.45]
			\draw[->,color=black!100] (0,-0.3) -- coordinate (y axis mid) (0,4.3); 

			\foreach \x in {0,...,4}{
				\draw[color=black!0] (-3pt,\x) -- (-3pt,\x) node[anchor=east,color=black] {$\x=l_\x$}; 
			}

			\foreach \x in {1,...,2}{
				\draw[-,color=black!40] (0,2*\x-1) -- coordinate (x axis mid) (13,2*\x-1); 
			}
			\foreach \x in {0,...,2}{
				\draw[-,color=black!20] (0,2*\x) -- coordinate (x axis mid) (13,2*\x); 
			}

			\draw [very thick, color=black!100] (0,1) -- (1,2) -- (2,3) -- (3,2) -- (4,1) -- (5,0) -- (6,1) -- (7,2) -- (8,3) -- (9,4) -- (10,3);

			\candle{12.5}{1}{3}{4}{0};
		\end{tikzpicture}
		\label{figStepsD}
		}

		\caption{Construction of smaller IPMS with the same CR (cf. Proof of Theorem~\ref{upperLimit} and Example~\ref{bspShorten}).}
		\label{figSteps}
		\end{figure}
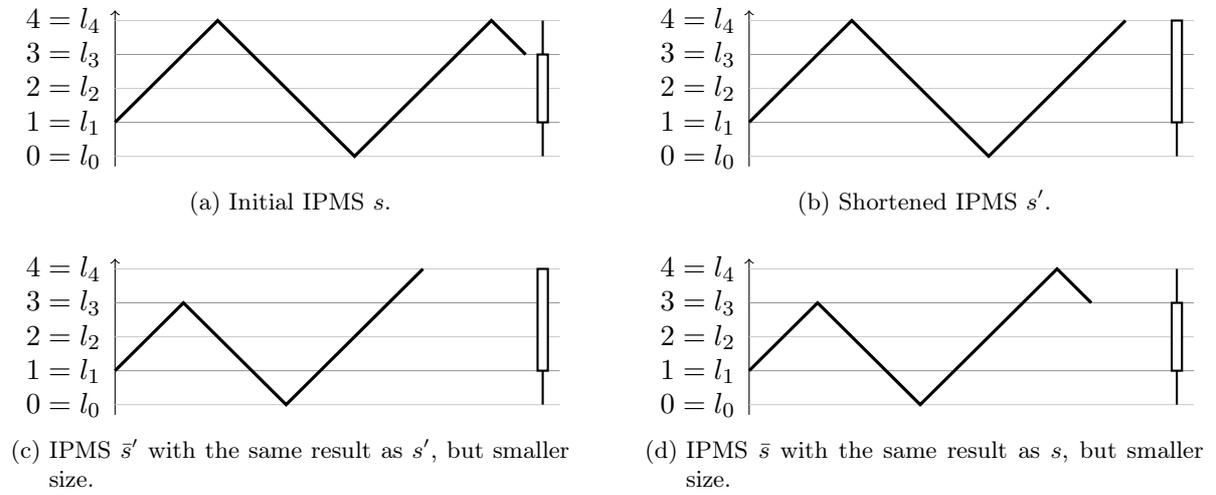
	\end{sample}


	\section{Conclusion}
	\label{sec:conclusion}

	In this work, we provided the utilities to test correctness of backtest engines for setups with at most one entry and one exit.
	By this many practical situations are covered, such as EnterLong/Short Limit/Stop with accompanying intra-period stop loss and target orders.

	By constructing all relevant intra-period price functions, which we could limit to finitely many intra-period model price series (IPMSs), and then running a reference backtest for IPMSs, we can obtain the correct result for all model candles.
	Comparing these results with the results of a given backtest engine on all model candles, we can decide correctness of the backtest engine under the assumption of stability under transformations by the results of Section~\ref{sec:correctness}.

	Many of those concepts can be generalized to more complex situations with more than one entry or exit.
	It remains to be shown how our results can be transferred to these setups.
	Furthermore, in that case there would no longer necessarily be unique worst cases and best cases.

	Another extension of our work could be the consideration of other order types, e.g. OCO-orders (i.e. ``one cancels other''), which would require a revision of the theory presented.


	\section*{Acknowledgement}
	Robert Löw was funded by an Undergraduate Fund Project 2015, RWTH Aachen.


\end{document}